\documentclass[pdftex,11pt,letterpaper]{extarticle}

\newcommand{\PAGENUMBERS}{yes}       % "yes" or "no"
\newcommand{\COLOR}{yes}

\newcommand{\comment}[1]{}
\newcommand{\onlyAbstract}{no}

\usepackage{amsmath,amssymb}
\usepackage{fullpage}

\usepackage{xspace}
\usepackage{color}
\usepackage{verbatim}
\usepackage{ifthen}

\usepackage{nicefrac}

\newcommand{\vone}{\mat{1}}
\newcommand{\mident}{\mat{I}}
\newcommand{\mleq}{\preccurlyeq}
\newcommand{\mgeq}{\succcurlyeq}

\newcommand{\poly}[1]{\textrm{poly}(#1)}

\usepackage{ifthen}
\ifthenelse{\equal{\PAGENUMBERS}{yes}}{%
\usepackage[nohead,
            left=0.85in,right=0.85in,top=0.75in,
            footskip=0.5in,bottom=1in,     % Room for page numbers
            columnsep=0.25in
            ]{geometry}
}{%
\usepackage[noheadfoot,left=0.75in,right=0.85in,top=0.75in,
            footskip=0.5in,bottom=1in,
            columnsep=0.25in
	    ]{geometry}
}

\usepackage[labelfont=bf]{caption}
\usepackage{enumitem}
\setlist{itemsep=0pt,parsep=0pt}             % more compact lists

\usepackage{fancyhdr}

\ifthenelse{\equal{\PAGENUMBERS}{yes}}{%
  \pagestyle{plain}
}{%
  \pagestyle{empty}
}

\usepackage[square,sort,comma,numbers]{natbib}

\usepackage{booktabs}
\usepackage{color}
\usepackage{colortbl}

\usepackage{txfonts}
\usepackage{mathabx}
\usepackage{ktmath}

\ifthenelse{\equal{\COLOR}{yes}}{%
  \usepackage[colorlinks,
  citecolor=blue,
  pagebackref=true
  ]{hyperref}%         % for online version
}{%
  \usepackage[pdfborder={0 0 0}]{hyperref}%  % for paper (B&W) version
}
\usepackage{url}

\definecolor{placeholderbg}{rgb}{0.85,0.85,0.85}

\usepackage{float}

\newcommand{\procName}[1]{\mbox{\texttt{#1}}\xspace}

\date{}
\title{Faster and Simpler Width-Independent Parallel Algorithms for Positive Semidefinite Programming}

\newcommand*{\email}[1]{\mbox{\small\texttt{#1}}}

\author{
Richard Peng
\\Georgia Tech
\footnote{Part of this work was done while at Carnegie Mellon University and while at M.I.T.}
\\\email{rpeng@cc.gatech.edu}
\and
Kanat Tangwongsan
\\Mahidol University International College
\footnote{Part of this work was done while at Carnegie Mellon University}
\\\email{kanat.tan@mahidol.edu}
\and
Peng Zhang
\\Georgia Tech
\\\email{pzhang60@gatech.edu}
}

\usepackage{float}
\usepackage[ruled,vlined,linesnumbered,noend,resetcount,algosection]{algorithm2e}
\SetNlSty{textsf}{}{:}
\DontPrintSemicolon
\SetKwComment{tcp}{$\rhd$\, }{}
\floatstyle{ruled}
\newfloat{algo}{tbp}{lop}[section]
\floatname{algo}{Algorithm}
\usepackage{tikz}

\usepackage{microtype}  % SPACE

\begin{document}

%\input{note}

% \begin{titlepage}
% \def\thepage{}
% \thispagestyle{empty}

% \date{}
\maketitle

\begin{abstract}
  This paper studies the problem of finding an $(1+\vareps)$-approximate solution
  to positive semidefinite programs.
  These are semidefinite programs in which all matrices in the constraints
  and objective are positive semidefinite and all scalars are non-negative.

  We present a simpler \NC parallel algorithm that on input with $n$ constraint
  matrices, requires $O(\frac{1}{\vareps^3} \log^3 n)$
  iterations, each of which involves only simple matrix operations and computing
  the trace of the product of a matrix exponential and a positive semidefinite
  matrix.
  Further, given a positive SDP in a factorized form, the total work of
  our algorithm is nearly-linear in the number of non-zero entries in the factorization.
  
\end{abstract}
%%% Local Variables:
%%% mode: latex
%%% TeX-master: "paper"
%%% End:

% Draft: notice
% ======================
% \begin{center}
% \vspace{-1cm}
% \bf{\huge{DRAFT}}
% \end{center}
% ======================

% \end{titlepage}
% \newpage

\ifthenelse{\equal{\onlyAbstract}{no}}{%
%\pagebreak
\section{Introduction}
\label{sec:intro}

\newcommand{\maxcut}{\textsc{MaxCut}\xspace}
\newcommand{\spcut}{\textsc{Sparsest Cut}\xspace}

Semidefinite programming (SDP), alongside linear programming (LP), is an
important tool in approximation algorithms, optimization, and discrete
mathematics.  In the context of approximation algorithms alone, it has emerged
as a key technique which underlies a number of impressive results that
substantially improve the approximation ratios.  To solve a semidefinite
program, algorithms from the linear programming literature such as Ellipsoid or
interior-point algorithms~\cite{GLS:book93} can be applied to derive near exact
solutions.  But they are often costly.  As a result, finding efficient
approximations to such problems is a critical step in making them more
practical.

From a parallel algorithms standpoint, both LPs and SDPs are \PCLASS-complete
even to approximate to any constant accuracy, suggesting that it is unlikely
that they have a polylogarithmic depth algorithm.  For linear programs, however,
the special case of positive linear programs, first studied by Luby and
Nisan~\cite{LubyN:STOC93}, has an algorithm that finds a
$(1+\vareps)$-approximate solution in $O(\poly{\tfrac1{\vareps}\log{n}})$
iterations.  This weaker approximation guarantee is still sufficient for
approximation algorithms (e.g., solutions to vertex cover and set cover via
randomized rounding), spurring interest in studying these problems in both
sequential and parallel contexts (see, e.g.,
\cite{LubyN:STOC93,PlotkinST95,GargK98,Young01:FOCS01,
  KoufogiannakisY:focs07,KoufogiannakisY09:PODC09}).

The importance of problems such as \maxcut and \spcut has led to the
identification and study of positive SDPs.
Our work is motivated by a result by Jain and Yao~\cite{JainY:focs11}
that gave the first positive SDP algorithm whose work and depth are
independent of the width parameter
(commonly known as width-independent algorithms).
As there are substantial differences in the analysis in this version compared to the
conference version~\cite{PengT12}, we address the relation between our paper
and other works in Section~\ref{subsec:related}.

We present a simple algorithm that offers guarantees similar to~\cite{JainY:focs11}
but has less work-depth complexity.
Each iteration of our algorithm involves only simple matrix
operations and computing the trace of the product of a matrix exponential and a
positive semidefinite matrix.
The input consists of an accuracy parameter $\vareps > 0$ and a positive
semidefinite program (PSDP) in the following standard primal form:
\begin{align}
  \label{eq:covering}
  \begin{array}{l l l}
    \text{Minimize} & \mat{C} \bigdot \mat{Y}\\
    \text{Subject to:} & \mat{A}_i \bigdot \mat{Y}  \geq b_i &\text{ for } i = 1, \dots, n\\
    &\mat{Y}  \mgeq \mat{0},
  \end{array}
\end{align}
where the matrices $\mat{C}, \mat{A}_1, \dots, \mat{A}_n$ are $m$-by-$m$
symmetric positive semidefinite matrices, $\bigdot$ denotes the pointwise dot
product between matrices (see Section \ref{sec:background}), and the scalars
$b_1, \dots, b_n$ are nonnegative reals.  This is a subclass of SDPs where the
matrices and scalars are ``positive'' in their respective settings.  We also
assume, as is standard, that the SDP has strong duality.  Our main result is as
follows:
\begin{theorem}[Main Theorem]
  Given a primal positive SDP involving $m \times m$ matrices with $n$
  constraints and an accuracy parameter $\vareps > 0$, there is an algorithm
  \procName{approxPSDP} that produces a $(1 + \vareps)$-approximation in
  $O(\frac{1}{\vareps^3}\log^{3} n)$ iterations, where
  each iteration involves computing matrix sums and a special primitive that
  computes $\exp(\mat{\Phi})\bigdot \mat{A}$ in the case when $\mat{\Phi}$ and
  $\mat{A}$ are both positive semidefinite.
\end{theorem}

The theorem quantifies the cost of our algorithm in terms of the number of
iterations. The work and depth bounds implied by this theorem vary with the
format of the input and how the matrix exponential is computed in each
iteration.  As we will discuss in Section~\ref{sec:exp}, with input given in a
suitable form, our algorithm runs in nearly-linear work and polylogarithmic
depth.

%%% Local Variables:
%%% mode: latex
%%% TeX-master: "paper"
%%% End:

\subsection{Related Work}

\label{subsec:related}

The first definition of positive SDPs was due to Klein and Lu~\cite{KleinL:stoc96},
who used it to characterize the \maxcut SDP. 
The \maxcut SDP can be viewed as a direct generalization of
positive (packing) LPs.
More recent work defines the notion of positive packing
SDPs~\cite{IyengarPS:siamjo11}, which captures problems such as \maxcut, sparse
PCA, and coloring; and the notion of covering SDPs~\cite{IyengarPS:swat10},
which captures the ARV relaxation of \spcut among others.
Works this area tend to focus on developing fast sequential algorithms for finding a
$(1+\vareps)$-approximation, leading to a series of  sequential
algorithms~(e.g.,~\cite{AroraHK:focs05,AroraK:stoc07,IyengarPS:siamjo11,IyengarPS:swat10}). 
 The iteration count for these algorithms, however, depends on the so-called
``width'' parameter of the input program or some parameter of the spectrum of
the input program. 
In some instances, the width parameter can be as large as
$\Omega(n)$, making it a bottleneck in the depth of direct parallelization.
This is even more so in the case of the \spcut SDP even though the problem has
been labeled as a covering SDP~\cite{IyengarPS:swat10}.

Over the past few years, width-independent algorithms for positive SDPs have
received much attention.  Jain and Yao gave the first result in this direction:
a polylog depth algorithm based on the first width independent linear
programming algoirthm by Luby and Nisan~\cite{LubyN:STOC93}.  Their algorithm is
based on updating the primal matrix.  This leads to an intricate analysis based
on carefully analyzing the eigenspaces of a particular matrix before and after
each update.  It takes $O( \frac{1}{\vareps^{13}} \log^{13} m \log n)$
iterations, each of which involves computing spectral decompositions using least
$\Omega(m^{\omega})$ work.

Our algorithm follows a different approach.
It updates the dual program in a way motivated by the width-independent
positive LP algorithm by Young~\cite{Young01:FOCS01}.
Concurrently, Jain and Yao~\cite{JainY:arxiv12} gave a similar algorithm for positive SDPs.
Their algorithm solves a class of SDPs which contains both packing and diagonal
covering constraints.  Since matrix packing conditions between diagonal matrices
are equivalent to point-wise conditions of the diagonal entries, these
constraints are closer to a generalization of positive covering LP constraints.
We believe that removing this restriction on diagonal packing matrices would
greatly widen the class of problems included in this class of SDPs and discuss
possibilities in this direction in Section~\ref{sec:concl}.

The convergence analyses of both of these routines go through the matrix
multiplicative weights update framework for solving semidefinite
programs~\cite{AroraK:stoc07}.
However, a crucial algebraic piece is missing for adapting the iteration
count bound from Young's algorithm.
Unlike scalar exponentials, matrix exponentials is not a monotonic function
under standard notions of matrix ordering such as the Loewner partial order.

Most recently, Allen-Zhu et~al.~gave the first rigorous adaptation of Young's
algorithm to positive SDPs~\cite{ALO15:arxiv}.
Motivated by this result, as well as its precursor in positive LPs~\cite{AO15}, 
we complete the analysis of our original algorithm.
 Specifically, we show that the dual generation scheme in these routines provide
 good bounds on the iteration count.
 This leads to a different analysis of Young's algorithm that omits phases.
 Our modified analysis is for a simplified pseudocode of the algorithm
 from~\cite{PengT12} that removes these phases.
 However, the phase-based version can be analyzed similarly.

Compared to the algorithm by Allen-Zhu et al.~\cite{ALO15:arxiv}, our analysis uses more
elementary linear algebraic techniques, and is closer to multiplicative weights update schemes.
We believe the dynamic bucketing method for obtaining better dependencies
on $\vareps$~\cite{WangMMR15} is also applicable to our analysis.
Finally, as many recent results on faster positive LP / SDP algorithms take optimization
based views, we believe our approach is also of independent interest.

%%% Local Variables:
%%% mode: latex
%%% TeX-master: "paper"
%%% End:

\subsection{Overview}
\label{subsec:overview}

We derive our algorithm by generalizing Young's algoirthm~\cite{Young01:FOCS01}.
In place of the ``soft max'' function for bounding the maximum of a set of
linear constraints, we use matrix exponential and the matrix multiplicative
weights update (MMWU) mechanism.  Moreover, besides standard operations on
(sparse) matrix, the only other primitive needed is the matrix dot product
$\exp(\mat{\Phi})\bigdot \mat{A}$, where $\mat{\Phi}$ and $\mat{A}$ are positive
semidefinite.

\bigskip

\noindent \textbf{Intuitions.} For intuition about packing SDPs and the matrix
multiplicative weights update method for finding approximate solutions, a useful
analogy of the decision problem is that of packing a (fractional) amount of
ellipses into the unit ball.  Figure \ref{fig:example2d} provides an example
involving $3$ matrices (ellipses) in $2$ dimensions.  Note that $\mat{A}_1$ and
$\mat{A}_2$ are axis-aligned; their sum is also axis-aligned in this case.  In
fact, positive linear programs in the broader context corresponds exactly to the
restriction of all ellipsoids being axis-aligned.  In this setting, the
algorithm of \cite{Young01:FOCS01} can be viewed as creating a penalty function
by weighting the length of the axises using an exponential function.  Then,
ellipsoids with sufficiently small penalty subject to this function have their
weights increased.

However, once we allow general ellipsoids such as $\mat{A}_3$, the resulting sum
will no longer be axis-aligned.  In this setting, a natural extension is to take
the exponential of the semimajor axises of the resulting ellipsoid instead,
leading to the matrix multiplicative weights scheme.  Our analysis then focuses
on showing that Young's algorithm still has a width-independent iteration count
in this setting.

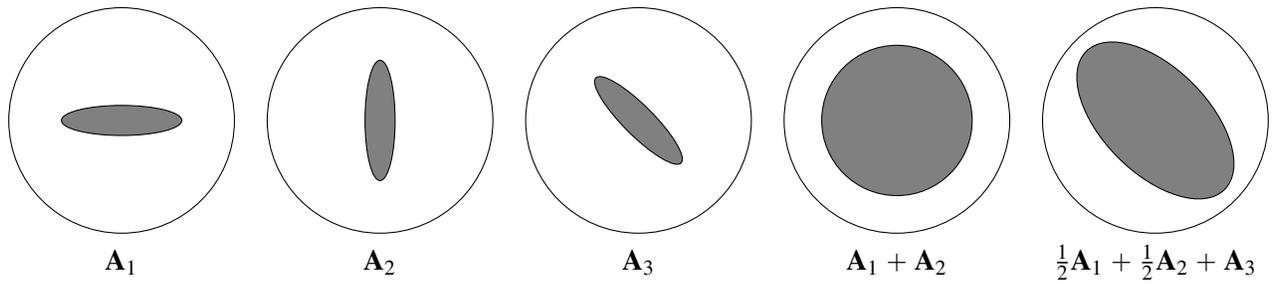
\begin{figure*}
\begin{tabular}{ccccc}
\begin{tikzpicture}
\draw[rotate=0] (0,0) circle (15mm and 15mm);

\filldraw[fill=gray,rotate=0] (0,0) circle (8mm and 2mm);
\end{tikzpicture}
&
\begin{tikzpicture}
\draw[rotate=0] (0,0) circle (15mm and 15mm);

\filldraw[fill=gray,rotate=0] (0,0) circle (2mm and 8mm);
\end{tikzpicture}
&
\begin{tikzpicture}
\draw[rotate=0] (0,0) circle (15mm and 15mm);

\filldraw[fill=gray,rotate=45] (0,0) circle (2mm and 8 mm);
\end{tikzpicture}
&
\begin{tikzpicture}
\draw[rotate=0] (0,0) circle (15mm and 15mm);
\filldraw[fill=gray,rotate=0] (0,0) circle (10mm and 10mm);
\end{tikzpicture}
&
\begin{tikzpicture}
\draw[rotate=0] (0,0) circle (15mm and 15mm);
\filldraw[fill=gray,rotate=45] (0,0) circle (7mm and 13 mm);
\end{tikzpicture}\\
$\mat{A}_1$ &
$\mat{A}_2$ &
$\mat{A}_3$ &
$\mat{A}_1 + \mat{A}_2$ &
$\frac{1}{2} \mat{A}_1 + \frac{1}{2} \mat{A}_2 + \mat{A}_3$
\end{tabular}
\caption{An instance of a packing SDP in 2 dimensions.}
\label{fig:example2d}
\end{figure*}

\bigskip

\noindent\textbf{Work and Depth.} 
We now discuss the work and depth bounds of our algorithm.  The main cost of
each iteration of our algorithm comes from computing the dot product between a
matrix exponential and a PSD matrix.  Like in the sequential setting
\cite{AroraHK:focs05, AroraK:stoc07}, we need to compute for each iteration the
product $\mat{A}_i \bigdot \exp(\mat{\Phi})$, where $\mat{\Phi}$ is some PSD
matrix.  The cost of our algorithm therefore depends on how the input is
specified. When the input is given prefactored---that is, the $m$-by-$m$
matrices $\mat{A}_i$'s are given as $\mat{A}_i = \mat{Q}_i\mat{Q}_i^\tr$ and the
matrix $\mat{C}^{-1/2}$ is given, then Theorem~\ref{thm:exp} can be used to
compute matrix exponential in $O(\frac1{\vareps^3}(m+q)\log n \log q \log
(\nicefrac{1}{\vareps}))$ work and $O(\frac1{\vareps}\log n \log q \log
(\nicefrac1{\vareps}))$ depth, where $q$ is the number of nonzero entries across
$\mat{Q}_i$'s and $\mat{C}^{-1/2}$.  This is because the matrix $\mat{\Phi}$
that we exponentiate has $\norm[2]{\mat{\Phi}} \leq O(\frac1\vareps\log n)$, as
shown in Lemma~\ref{eq:par-spec-bound}.
Therefore, as a corollary to the main theorem, we have the following cost
bounds:
\begin{corollary}
  The algorithm \procName{approxPSDP} has in $\otilde(\frac{1}{\epsilon^6}(n + m + q))$ work and
  $O(\frac{1}{\epsilon^4}\log^{O(1)} (n+m+q))$ depth.
\end{corollary}

If, however, the input program is not given in this form, we can add a
preprocessing step that factors each $\mat{A}_i$ into $\mat{Q}_i\mat{Q}_i^\tr$
since $\mat{A}_i$ is positive semidefinite.  In general, this preprocessing
requires at most $O(m^{4})$ work and $O(\log^3 m)$ depth using standard parallel
QR factorization~\cite{JaJa:book92}.  Furthermore, these matrices often have
certain structure that makes them easier to factor.  Similarly, we can factor
and invert $\mat{C}$ with the same cost bound, and can do better if it also has
specialized structure.

%%% Local Variables:
%%% mode: latex
%%% TeX-master: "paper"
%%% End:

\section{Background and Notation}
\label{sec:background}

We review notation and facts that will prove useful later in the paper. We write
$\otilde(f(n))$ to mean $O(f(n)\polylog(f(n)))$.  Throughout the paper, assume
$1/\vareps \in \poly{n}$, so $\log(1 / \vareps) = O(\log{n})$.  For larger
$1/\vareps$, SDP solvers with $\log(1 / \vareps)$
dependencies~\cite{Boyd&Vandenberghe:2004} have better performance.

% Because our
% algorithm has $\poly{1 / \vareps}$ dependencies, 

% Our algorithms have $\poly{1 / \vareps}$
% dependencies. Since routines for solving semidefinite programs with
% $\log(1 / \vareps)$ dependencies~\cite{Boyd&Vandenberghe:2004} have better
% performance when $1 / \vareps$ is large, we will assume
% $1 / \vareps \in \poly{n}$, aka $\log(1 / \vareps) = O(\log{n})$ throughout our
% presentation.

% \kt{How does it matter that there are poly-time routines for solving
% semidefinite programs? I'm just confused by the way we phrase the reasoning.}
% \rp{tried to fix}

\bigskip

\subsection{Linear Algebraic Notation}

\noindent\textbf{Matrices and Positive Semidefiniteness.}
\newcommand{\zz}{\textbf{z}}
Unless otherwise stated, we will deal with real symmetric matrices in
$\R^{m \times m}$. A symmetric matrix $\mat{A}$ is positive
semidefinite, denoted by $\mat{A} \mgeq \mat{0}$ or $\mat{0} \mleq
\mat{A}$, if for all $\zz \in \R^m$, $\zz^\tr \mat{A} \zz \geq 0$.
Equivalently, this means that all eigenvalues of $\mat{A}$ are
non-negative and the matrix $\mat{A}$ can be written as
\[
\mat{A} = \sum_{i} \lambda_i \mat{v}_i\mat{v}_i^\tr,
\]
where $\mat{v}_1, \mat{v}_2, \dots, \mat{v}_m$ are the eigenvectors of
$\mat{A}$ with eigenvalues $\lambda_1 \geq \dots \geq \lambda_m$ respectively.
We will use $\lambda_1(\mat{A})$,
$\lambda_2(\mat{A}), \dots, \lambda_m(\mat{A})$ to represent the
eigenvalues of $\mat{A}$ in decreasing order and
also use $\lambda_{\max}(\mat{A})$ to denote $\lambda_1(\mat{A})$.
Notice that positive semidefiniteness induces a partial ordering on
matrices. We write $\mat{A} \mleq \mat{B}$ if $\mat{B} -
\mat{A} \mgeq \mat{0}$.

The trace of a matrix $\mat{A}$, denoted $\trace{\mat{A}}$, is the sum
of the matrix's diagonal entries: $\trace{\mat{A}} = \sum_i A_{i,i}$.
Alternatively, the trace of a matrix can be expressed as the sum of
its eigenvalues, so $\trace{\mat{A}} = \sum_i \lambda_i(\mat{A})$.
Furthermore, we define
\[
\mat{A}\bigdot \mat{B} = \sum_{i,j} A_{i,j}B_{i,j} = \trace{\mat{A}\mat{B}}.
\]
It follows that $\mat{A}$ is positive semidefinite if and only if
$\mat{A}\bigdot \mat{B} \geq 0$ for all PSD $\mat{B}$.

\bigskip

\noindent\textbf{Matrix Exponential.} Given an $m \times m$ symmetric
positive semidefinite matrix $\mat{A}$ and a function $f\!: \R \to \R$, we
define
\[
f(\mat{A}) = \sum_{i=1}^m f(\lambda_i) \mat{v}_i\mat{v}_i^\tr,
\]
where, again, $\mat{v}_i$ is the eigenvector corresponding to the
eigenvalue $\lambda_i$.  It is not difficult to check that for
$\exp(\mat{A})$, this definition coincides with $\exp(\mat{A}) =
\sum_{i\geq 0} \frac1{i!}\mat{A}^i$.

Our algorithm relies on a matrix multiplicative weights (MMW) algorithm, which
can be summarized as follows.  For a fixed $\vareps_0 \leq \frac12$ and
$\mat{W}^{(1)} = \mat{I}$, we play a ``game'' a number of times, where in
iteration $t = 1, 2, \dots$, the following steps are performed:

%\hrule
\begin{enumerate}[topsep=2pt]
\item Produce a ``probability'' matrix $\mat{P}^{(t)} =
  \mat{W}^{(t)}/\trace{\mat{W}^{(t)}}$;

\item Incur a gain matrix $\mat{M}^{(t)}$; and
\item Update the weight matrix as \[ \mat{W}^{(t+1)} =
  \exp(\vareps_0\sum_{t' \leq t} \mat{M}^{(t')}).\]
\end{enumerate}

%\hrule

Like in the standard setting of multiplicative weights algorithms, the gain
matrix is chosen by an external party, possibly adversarially.  In our
algorithm, the gain matrix is chosen to reflect the step we make in the
iteration.  Arora and Kale~\cite{AroraK:stoc07} shows that the MMW algorithm has
the following guarantees (restated for our setting):
\begin{theorem}[\cite{AroraK:stoc07}]
  \label{thm:mmwu}
  For $\vareps_0 \leq \frac12$, if $\mat{M}^{(t)}$'s are all PSD and
  $\mat{M}^{(t)} \mleq \mident$, then after $T$ iterations,
  \begin{align}
    (1+\vareps_0) \sum_{t=1}^T \mat{M}^{(t)} \bigdot \mat{P}^{(t)}
    &\geq \lambda_{\max}\left(\sum_{t=1}^T \mat{M}^{(t)}\right) - \frac{\ln n}{\vareps_0}.
  \end{align}
\end{theorem}

\begin{comment}
For completeness, we give a proof of this theorem in Appendix~\ref{sec:mmwu}.
\end{comment}

\subsection{Reduction to Bounded Decision Version}
\label{subsec:reduction}

Our algorithm works with normalized primal/dual programs
shown in Figure~\ref{fig:norm-pd-sdp}.

\begin{figure*}[ht]
\begin{align}
  \label{eq:simplified-sdp}
  \begin{array}{c | c}
    \underline{\textit{Primal}~(Covering) } & \underline{\textit{Dual}~(Packing)}\\[0.5em]
    \begin{array}{lrl}
      \mbox{Minimize} &  \trace{\mat{Y}}\\
      \mbox{Subject to:} & \mat{A}'_i \bigdot \mat{Y} & \geq 1 \qquad \text{ for } i = 1,\dots, n\\
      & \mat{Y} & \mgeq 0
    \end{array}
    &
    \begin{array}{lr l}
      \mbox{Maximize} & \vone^\tr \xx \\
      \mbox{Subject to:} & \sum_{i=1}^n x_i \mat{A}'_i & \mleq \mident \\
      &\xx & \geq \bvec{0}.
    \end{array}
  \end{array}
\end{align}
\caption{Normalized primal/dual positive SDPs. The symbol $\mident$ represents
  the identity matrix.}
\label{fig:norm-pd-sdp}
\end{figure*}

By using binary search and appropriately scaling the input program, such an SDP
can be approximated using the following decision problem:
%%
%% We solve a normalized SDP by resorting to
%%an algorithm for its decision version and binary search.  In particular, we
%%design an algorithm with the property that given a goal value $\obj$, either
%%find a dual solution $\xx \in \R^+_n$ to (\ref{eq:simplified-sdp}-D) with
%%objective at least $(1-\vareps)\obj$, or a primal solution $\mat{Y}$ to
%%(\ref{eq:simplified-sdp}-P) with objective at most $\obj$.
\begin{quote}
  $\vareps$-\textbf{Decision Problem:} Find either an $\xx \in \R^+_n$ (a dual
  solution) such that
  \[ \norm[1]{\xx} \geq 1 - \vareps \text{ and } \sum_{i=1}^n x_i\mat{A}_i
  \mleq \mident \] or a PSD  matrix $\mat{Y}$ (a primal solution) such that
  \[
    \trace{\mat{Y}} = 1 \text { and } \forall i, \mat{A}_i \bigdot \mat{Y} \geq 1.
  \]
\end{quote}
This reduction can also ensure bounded trace on all $\mat{A}_i$'s.  The
following lemma summarizes key properties of such a reduction:
\begin{lemma}\label{lem:trace-limit}
  For $0 < \vareps < 1$, a positive packing semidefinite program can be
  approximated to a relative error to $\vareps$ using $O(\log{n})$ calls to the
  $\vareps$-decision problem.  Furthermore, each $\mat{A}_i$ supplied to the
  decision problem has $\trace{\mat{A}_i} \leq O(n^3)$.
% We can solve a positive packing semidefinite program to a relative
% error of $\vareps$ by solving a sequence
% of $O(\log{n})$ decision problems with error $\vareps$.
% Furthermore, in each of these problems we can ensure that
% $\trace{\mat{A}_i} \leq O(n^3)$.
\end{lemma}

The reduction is common to most positive linear and semidefinite program
solvers~\cite{JainY:focs11,AO15,ALO15:arxiv}; we briefly sketch the idea for
completeness.
\begin{proof} (Sektch) First, transform to the normalized form by ``dividing
  through'' by $\mat{C}$ (see Appendix~\ref{sec:norm-posit-sdps}).  Since
  $\trace{ \sum_{i} x_{i} \mat{A}_i}$ is within a factor of $n$ of the maximum
  eigenvalue of $ \sum_{i} x_{i} \mat{A}_i$,
  $\frac{1}{\min_{i} \trace{\mat{A}_i}}$ gives a value that is within a factor
  $n$ of the optimum.  Therefore, the optimization version can be solved by
  binary searching on the objective a total of at most
  $O(\log(\frac{n}{\vareps}))$ iterations.

  In each of these decision instances, we can rescale $\mat{A}_i$'s so that the
  threshold in question is $1$.  In this setting, the total sum of $x_i$ with
  $\trace{\mat{A}_i} \geq n^3$ is at most $1 / n$.  By not using these
  $x_i$'s, the optimum changes by an additive value of less than $\vareps$.
\end{proof}

%%% Local Variables:
%%% mode: latex
%%% TeX-master: "paper"
%%% End:

\section{Solving Positive SDPs}
\label{sec:parallel-packing}

In this section, we describe a parallel algorithm for solving the decision
version of positive packing SDPs, inspired by Young's algorithm for positive
LPs.  The following theorem presents the guarantees of our algorithm.
\begin{theorem}
\label{thm:decisionMain}
Let $0 < \vareps < 1$. There is an algorithm $\procName{decisionPSDP}$ that
given a positive SDP, solves the $\vareps$-decision problem in
$O(\vareps^{-3}\log^2 n)$ iterations.
\end{theorem}

Presented in Algorithm~\ref{algo:parpacking} is an algorithm that we will show
to satisfy the theorem.  

%% import the algorithm from a different file
\begin{algorithm}
\caption{Parallel Packing SDP algorithm, $\procName{decisionPSDP}$}
\label{algo:parpacking}
\tcp{Define $K = \frac1{\vareps}(1 + \ln n)$,
  $\alpha = \frac{\vareps/K}{1+10\vareps}$, and
  $R = \frac{32}{\vareps\alpha}\ln n$}
Initialize $t \gets 0$ and $x^{(0)}_i \gets \frac{1}{n\cdot \trace{\mat{A}_i}}$\;
%$\mat{\Psi}^{(0)} \gets \sum_{i = 1}^n x^{(0)}_i \mat{A}_i$, and $t \gets 0$\;
\While{$\norm[1]{\xx^{(t)}} \leq K$ \textbf{ and } $t < R$}{ $t \gets t + 1$\;
  Compute matrix exponential $\mat{W}^{(t)} \gets \sum_{i = 1}^n x^{(t-1)}_i \mat{A}_i$,\;
%  and corresponding matrix $\mat{P}^{(t)} = \nicefrac{\mat{W}^{(t)}}{\trace{\mat{W}^{(t)}}}$\;
  %
  Identify coordinates to update (in parallel),
%  $B^{(t)} \gets \{ i \in [n] : \mat{P}^{(t)} \bigdot \mat{A}_i \leq (1 + \vareps)\}$\;
  $B^{(t)} \gets \{ i \in [n] : \mat{W}^{(t)} \bigdot \mat{A}_i \leq (1 + \vareps) \trace{\mat{W}^{(t)}}\}$\; %
%  \rp{I'm trying to rewrite these to have the most succinct pseudocode possible, then reintroduce stuff in analysis}
% \kt{that's cool}

  Update $\xx^{(t)} \gets \xx^{(t - 1)} + \alpha \cdot \xx_{B}^{(t - 1)}$\;
}
\If{$\norm[1]{\xx^{(t)}} > K$}{
%%  $\xx^* \gets \frac1{(1+10\vareps)K}\xx^{(T)}$\;
  \Return{$\widehat{\xx} = \frac1{(1+10\vareps)K}\xx^{(t)}$ as a dual solution}\; 
%\rp{* usually denotes optimum solution}
%\kt{agreed, why did we use * in the first place? :(}
}
\Else{
  \Return{$\overline{\mat{Y}} = \tfrac{1}{t} \sum_{\tau=1}^t \mat{W}^{(\tau)}/\trace{\mat{W}^{(\tau)}}$ as a primal solution}\;
}
\end{algorithm}

%%% Local Variables:
%%% mode: latex
%%% TeX-master: "paper"
%%% End:

%%%%%%%%%%%%

The algorithm is a multiplicative weights update routine,
which proceeds in several rounds.
The starting solution is $\xx^{(0)}_i = \frac{1}{n \trace{\mat{A}_i}}$.
This solution is chosen to be small so that $\sum_i x_i^{(0)}\mat{A}_i \mleq \mat{I}$,
hence respecting the dual constraint.
Each subsequent update is a multiple of the current solution, so this
$\xx^{(0)}$ is also chosen to ensure that these updates make rapid progress.

In each iteration following that, the algorithm computes
\[
\mat{W}^{(t)} = \exp\pparen{\sum_i x^{(t-1)}_i \mat{A}_i}.
\]
Our presentation follows the multiplicative weights update framework from
Arora-Kale~\cite{AroraK:stoc07} and Kale~\cite{Kale:thesis07}.  Several
intermediate variables are helpful for further discussion.  Define the
cumulative sum corresponding to $\xx^{(t)}$ as
\begin{align}
  \mat{\Psi}^{(t)} \eqdef \sum_{i=1}^n x_i^{(t)}\mat{A}_i,
\end{align}
This allows us to write the exponential as
$\mat{W}^{(t)} = \exp ( \mat{\Psi}^{(t)} )$.
We will also specifically define the 
probability matrix by which we use to pick the update coordinates:
\begin{align}
  \mat{P}^{(t)} &\eqdef \mat{W}^{(t)}/\trace{\mat{W}^{(t)}}.
\end{align}
This matrix is easier to work with because it has trace $1$:
\begin{equation}
  \trace{\mat{P}^{(t)}} = \trace{\mat{W}^{(t)}/\trace{\mat{W}^{(t)}}}
          = \trace{\mat{W}^{(t)}} / \trace{\mat{W}^{(t)}}
        = 1 \label{eq:trace-of-P}
\end{equation}
Using this ``probability'' matrix $\mat{P}^{(t)}$, the algorithm identifies
which $\xx$ coordinates to update.  These are the coordinates $x_i$'s for which
their contributions with respect to the $\mat{A}_i$'s are still
small---$\mat{P}^{(t)} \bigdot \mat{A}_i \leq 1 + \vareps$.  Each of these $\xx$
coordinates will be incremented by $\delta^{(t)} \eqdef \alpha \cdot \xx_i$,
where $\alpha = \frac{\vareps/K}{1 + 10\vareps}$.  Therefore, in terms of
$\delta^{(t)}$, we have
\[
  \xx^{(t)} = \xx^{(0)} + \sum_{\tau=1}^t \delta^{(\tau)} %\label{eq:xx-vs-delta}\\
\]

The main loop in Algorithm \ref{algo:parpacking} terminates when
$\norm[1]{\xx^{(t)}} > K$ or the number of iterations exceeds a preset threshold
$R = O(\vareps^{-3}\log^2 n)$.  For a desired accuracy parameter $\vareps > 0$,
we set $K$ to $O( \frac1\vareps(1 + \ln n) )$ so that the $\ln n$ additive term
from Theorem~\ref{thm:mmwu} can be absorbed into the relative error.  This
additive term comes from the starting point $\xx^{(0)}$.

The choice of $\alpha$ may seem mysterious at this point; it is chosen to
prevent us from taking a step that is too big from the current solution while
still making substantial progress.  It ensures that
\begin{enumerate}
\item The update has small width, aka.
$\sum_{i} \delta_i^{(t)} \mat{A}_i \mleq \vareps \mident$, and
\item We cannot overshoot by much when exiting from the while loop,
$\vone^\tr \delta^{(t)} \leq \vareps$.
\end{enumerate}

%%We will bound the approximation guarantees and analyze the cost of the
%%algorithm.  Before we start, we will need some notation and definitions.  An
%%easy induction gives that the quantities that we track across the iterations of
%%Algorithm~\ref{algo:parpacking} satisfy the following relationships:
%%\begin{align}
%%  \trace{\mat{P}^{(t)}} &= \trace{\mat{W}^{(t)}/\trace{\mat{W}^{(t)}}} \nonumber\\
%%        &= \trace{\mat{W}^{(t)}} / \trace{\mat{W}^{(t)}}
%%        = 1 \label{eq:trace-of-P} \\
%%%  \displaybreak
%%  \mat{M}^{(t)} &\eqdef \frac1{\vareps}\sum_{i=1}^n \delta_i^{(t)} \mat{A}_i
%%	\qquad \text{when $t \geq 1$}\\
%%  \allowdisplaybreaks
%%%	= \vareps \sum_{\tau=0}^t \mat{M}^{(\tau)}
%%%\qquad \text{When $\xx^{(0)} = 0$ \richard{Should we just specify this as part of the algorithm?}}
%%\end{align}

To bound the approximation guarantees and the cost of this algorithm, we reason
about the spectrum of $\mat{\Psi}^{(t)}$ and the $\ell_1$ norm of the vector
$\xx^{(t)}$ as the algorithm executes.  Since the coordinates of our vector
$\xx^{(t)}$ are always nonnegative, we note that $\norm[1]{\xx^{(t)}} =
\vone^\tr \xx^{(t)}$ and use either notation as convenient.  

Our analysis of $\procName{decisionPSDP}$ and in turn our proof
of Theorem~\ref{thm:decisionMain} hinge on two complementary
processess:
show that after $R$ steps, $\overline{\mat{Y}}$ is indeed a feasible primal solution,
or that if the the algorithm terminates because $\norm[1]{\xx^{(t)}} > K$,
then $\xxhat$ is a feasible dual solution.
In the latter case, we have
\begin{equation}
\norm[1]{\widehat{\xx}} = \frac{1}{(1+10\vareps)K}\norm[1]{\xx^{(t)}}
\geq \frac{K}{(1+10\vareps)K} \geq 1 - 10 \vareps
\label{eq:ell-one-dual-sol}
\end{equation}
For this $\xxhat$ to be a dual solution, we still need to show that it satisfies
$\sum_i \widehat{\xx}_i \mat{A}_i \mleq \mident$.

\subsection*{Bounding The Spectrum}
%\noindent\textbf{Spectrum Bounds.}

Let $T$ be the final iteration count (i.e., the final $t$).  To meet the
requirement above, we only need to show that
$\frac1{(1+10\vareps)K}\mat{\Psi}^{(T)} \mleq \mident$.  We prove the following
spectrum bound:
\begin{lemma}[Spectrum Bound]
  \label{lem:par-spec-bound}
  For every $t=0,\dots, T$,
  \begin{equation}
  \mat{\Psi}^{(t)} = \sum_{i=1}^n x^{(t)}_i \mat{A}_i \mleq (1+10\vareps)K\mident.
  \label{eq:par-spec-bound}
\end{equation}
\end{lemma}
We prove this lemma by induction on the iteration number, resorting to
properties of the MMW algorithm (Theorem~\ref{thm:mmwu}), which relates the
final spectral values to the ``gain'' derived at each intermediate step.

To proceed, we will need a few facts about the algorithm (their proofs follow
after the proof of the lemma). Claim~\ref{claim:par-bc-lmax} shows that the
initial matrix (i.e., $t = 0$) satisfies the bound.
Claim~\ref{claim:par-step-size} claim quantifies the gain in each step as a
function of the $\ell_1$-norm change we make in that step.
Claim~\ref{claim:totalx} bounds the $\ell_1$-norm of $\xx$.

\begin{claim}
  \label{claim:par-bc-lmax}
  $\mat{\Psi}^{(0)} \mleq \mident$. In other words,
  $\lambda_{\max}\pparen{\mat{\Psi}^{(0)}} = \lambda_{\max}\left(\sum_{i=1}^n
    x_i^{(0)} \mat{A}_i\right) \leq 1$.
\end{claim}

\begin{claim}
  \label{claim:par-step-size}
  For $t = 1, \dots, T$,
  \begin{equation}
    \mat{M}^{(t)} \bigdot \mat{P}^{(t)} \leq \frac{(1+\vareps)}{\vareps} \cdot \norm[1]{\delta^{(t)}}.
  \end{equation}
\end{claim}

\begin{claim}
  \label{claim:totalx}
  For $t=1,\dots, T$,
  \[
  \norm[1]{\xx^{(t)}} \leq (1 + \vareps)K
  \]
\end{claim}

Now for any iteration $t \leq T$, we can rewrite $\mat{\Psi}^{(t)}$ as
\[
\mat{\Psi}^{(t)} = \sum_{i=1}^n x^{(0)}_i \mat{A}_i + \sum_{\tau=1}^t
\sum_{i=1}^n\delta^{(\tau)}_i \mat{A}_i =
\sum_{i=1}^n x^{(0)}_i \mat{A}_i + \vareps \sum_{\tau=1}^t \mat{M}^{(\tau)},
\]
so
\begin{equation}
  \lambda_{\max}(\mat{\Psi}^{(t)}) \leq 
       \lambda_{\max}\left(
         \sum_{i=1}^n x^{(0)}_i \mat{A}_i
       \right) + 
  \vareps \cdot \lambda_{\max}\left(
       \sum_{\tau=1}^t \mat{M}^{(\tau)}
  \right) 
  \leq 1 + \vareps\cdot \lambda_{\max}\left(
       \sum_{\tau=1}^t \mat{M}^{(\tau)}
  \right) 
  \label{eq:lmax-at-t}
\end{equation}
since both sums yield positive semidefinite matrices and the $\lambda_{\max}$ of
the first sum is at most $1$ by Claim~\ref{claim:par-bc-lmax}.

\medskip

\begin{proofof}{Lemma~\ref{lem:par-spec-bound}}
  We will prove \eqref{eq:par-spec-bound} by (strong) induction on $t$.  The
  base case of $t = 0$ is true by Claim~\ref{claim:par-bc-lmax}.  For a given
  $t$, if we inductively assume that
  $\mat{\Psi}^{(\tau)} \mleq (1 +10\vareps)K\mident$ for all $\tau < t$, then
  for each $1 \leq \tau < t$,
  \begin{align*}
    \mat{M}^{(\tau)} &= \frac1{\vareps}\sum_{i=1}^n \delta_i^{(\tau)} \mat{A}_i\\
                     &\mleq \frac{\alpha}{\vareps}\sum_{i=1}^n x_i^{(\tau-1)} \mat{A}_i \\
                     &= \frac{\vareps/K}{\vareps(1+10\vareps)}\sum_{i=1}^n
                       \mat{\Psi}^{(\tau-1)} \\
                     &\mleq \frac{\vareps/K}{\vareps(1+10\vareps)}
                       (1+10\vareps)K\mident \mleq \mident.
  \end{align*}
  This makes Theorem~\ref{thm:mmwu} applicable, which gives
  \begin{align*}
    \vareps\cdot{}\lambda_{\max}\left(\sum_{\tau=1}^t \mat{M}^{(\tau)}\right)
    &\leq   \vareps(1+\vareps) \sum_{\tau=1}^t \mat{M}^{(\tau)} \bigdot \mat{P}^{(\tau)}
      + \ln n \\
    &\leq  \vareps(1+\vareps) \sum_{\tau=1}^t \frac{(1+\vareps)}{\vareps} \cdot \vone^\tr{}\delta^{(\tau)} + \ln n &&\text{by \textsf{Claim~\ref{claim:par-step-size}}}\\
    &\leq (1+\vareps)^2 \vone^\tr{} x^{(t)} + \ln n && \text{by definition of $\delta^{(t)}$}\\
    &\leq (1+\vareps)^3 K + \ln n && \text{by \textsf{Claim~\ref{claim:totalx}}}
  \end{align*}
  Plugging this into \eqref{eq:lmax-at-t} yields
  \begin{align*}
    1+ \ln n + (1+\vareps)^3 K \leq \vareps K + (1+\vareps)^3 K,
  \end{align*}
  which allows us to conclude that
  $\mat{\Psi}^{(t)} \mleq (1+10\vareps)K\mident$, as desired.
\end{proofof}

\medskip

It remains to show the claims about the algorithm
utilized in the above proof.

\begin{proof}[of~Claim~\ref{claim:par-bc-lmax}]
  Our choice of $\xx^{(0)}$ guarantees that for all $i=1,\dots, n$,
  \[
  x_i^{(0)}\mat{A}_i = \frac{1}{n\trace{\mat{A}_i}} \mat{A}_i
  \mleq \frac{1}n \mident.
  \]
  Summing across $i=1,\dots, n$ gives the desired bound.
\end{proof}

\begin{proof}[of Claim~\ref{claim:par-step-size}]
  Consider that by definition,
  \begin{align*}
    \mat{M}^{(t)} \bigdot \mat{P}^{(t)}
    &= \frac1{\vareps}\left(\sum_{i=1}^n \delta_i^{(t)} \mat{A}_i\right)\bigdot \mat{P}^{(t)} = \frac1{\vareps}\left(\sum_{i \in B^{(t)}} \delta_i^{(t)} \mat{A}_i\bigdot \mat{P}^{(t)}\right)
  \end{align*}

  Now every $i \in B^{(t)}$, though $B^{(t)}$ can be empty, has the property
  that $\mat{A}_i\bigdot \mat{P}^{(t)} \leq (1+\vareps)$, so
  \begin{equation*}
    \mat{M}^{(t)} \bigdot \mat{P}^{(t)} \leq
    \frac{1+\vareps}{\vareps}\sum_{i\in B^{(t)}} \delta_i^{(t)} 
    \leq \frac{1+\vareps}{\vareps} \norm[1]{\delta^{(t)}},
  \end{equation*}
  which completes the proof.
\end{proof}

\begin{proof}[of~Claim~\ref{claim:totalx}]
  The condition of the \textbf{while}-loop ensures that for $t < T$ (i.e., prior
  to the final iteration), $\norm[1]{\xx^{(t)}} \leq K$.  For iteration $T$, we
  know that
  $\norm[1]{\xx^{(T)}} = \norm[1]{\xx^{(T-1)}} + \norm[1]{\delta^{(T)}}$ because
  $\delta^{(T)} \in \R^n_+$.  By our choice of $\alpha$, we know that
  $\alpha \leq \vareps$ and therefore
  $\norm[1]{\delta^{(T)}} = \alpha \norm[1]{\xx_{B^{(T)}}^{(T-1)}} \leq \vareps
  K$.
  Substituting this into the equation above gives
  $\norm[1]{\xx^{(T)}} \leq (1 + \vareps)K$, which proves the claim.
\end{proof}

It remains to examine the case where the algorithm returns a primal solution:
Equation~\eqref{eq:trace-of-P} gives
\[
  \trace{\overline{\mat{Y}}} = \frac{1}{T}\sum_{\tau=1}^{T} \trace{\mat{P}^{(\tau)}} = 1,
\]
Furthermore, this $\overline{\mat{Y}}$ satisfies the primal constraints:
\begin{lemma}
\label{lem:primalSolution}
If $\norm[1]{\xx^{(T)}} \leq K$---i.e, the algorithm exits the
\textbf{while}-loop because it reaches $R$ iterations---then for all
$i = 1, \dots, n$, $\mat{A}_i \bigdot \overline{\mat{Y}} \geq 1$.
\end{lemma}

\begin{proof}
  Assume for a contradiction that there is an $i \in [n]$ such that
  $\mat{A}_i \bigdot \overline{\mat{Y}} < 1$.  This means
  \[
  \frac{1}{T}\sum_{\tau=1}^T \mat{P}^{(\tau)} \bigdot \mat{A}_i < 1.
  \]
  Let $U = \{ \tau : \mat{P}^{(\tau)} \bigdot \mat{A}_i < 1 + \vareps\}$ be the
  iterations in which the $i$-th coordinate of $\xx$ is incremented.  By
  Markov's inequality, the number of such iterations is bounded by
  $|U| < \frac{\vareps}{1+\vareps}T$.  But then, every time the $i$-th
  coordinate changes, it increases by a factor of $1 + \alpha$, so
  \begin{align*}
  x_i^{(T)} > x_i^{(0)}(1+\alpha)^{\frac{\vareps}{1+\vareps}T} 
    >  x_i^{(0)}\exp\left(
    \frac{\alpha}{2}\cdot \frac{\vareps{}T}{1+\vareps}
    \right)
  \end{align*}
  Because $\norm[1]{\xx^{(T)}} \leq K$, the algorithm exits the
  \textbf{while}-loop with $T = R$. Therefore, for $0 < \vareps < 1$,
  \begin{align*}
    x_i^{(T)} > x_i^{(0)}\exp\left(
    \frac{\alpha\vareps}{2(1+\vareps)}\cdot R
    \right) 
    > x_i^{(0)}\cdot n^8
    > \frac{n^8}{n \trace{\mat{A}_i}} > \Omega(n^4)
  \end{align*}
  as $\trace{\mat{A}_i} \leq O(n^3)$ by Lemma~\ref{lem:trace-limit}.  This is a
  contradiction to $\norm[1]{\xx^{(T)}} \leq K$, which proves the lemma.
\end{proof}

We will now piece everything together:

\begin{proofof}{Theorem~\ref{thm:decisionMain}}
The algorithm terminates after at most $R$ iterations.
Notice that $B^{(t)}$ may be empty in some iterations but this does
not harm the algorithm nor the proof. It is standard to check that
$R = O(\vareps^{-3} \log^2 n)$.
If we do run this many iterations, then Lemma~\ref{lem:primalSolution}
gives that we terminate with a primal solution.

Otherwise, Lemma~\ref{lem:par-spec-bound} gives that at any point in
the algorithm, the solution vector $\xx^{(t)}$ satisfies
$\sum_i x^{(t)}_i \mat{A}_i \mleq (1 + 10\vareps)K\mident$.
Together with Equation \eqref{eq:ell-one-dual-sol}, we know that any
$\xxhat$ returned $\norm[1]{\xx^*} \geq 1 - 10\vareps$ and
\[
\sum_i \xxhat _i \mat{A}_i = \tfrac{1}{(1+10\vareps)K}\sum_i x^{(t)}_i \mat{A}_i
\mleq \mident.
\]
Thus, $\xxhat$ is indeed a dual solution with value at least $1 - 10\vareps$.
Replacing $\vareps$ with $\vareps/10$ then meets the requirements of
the decision problem.
\end{proofof}
%\end{proof}

%%% Local Variables:
%%% mode: latex
%%% TeX-master: "paper"
%%% End:

%\input{sequential-packing}
\section{Matrix Exponential Evaluation}

%Evaluation of $\exp{(\mat{\Phi})} \bigdot \mat{A}_i$}
\label{sec:exp}

We describe a fast algorithm for computing the matrix dot product of a positive
semidefinite matrix and the matrix exponential of another positive semidefinite
matrix.

\begin{theorem}
\label{thm:exp}
There is an algorithm $\procName{bigDotExp}$ that when given a $m$-by-$m$ matrix
$\mat{\Phi}$ with $p$ non-zero entries, $\kappa \geq \max\{1,
\norm[2]{\mat{\Phi}} \}$, and $m$-by-$m$ matrices $\mat{A}_i$ in factorized form
$\mat{A}_i = \mat{Q}_i\mat{Q}_i^\tr$ where the total number of nonzeros across
all $\mat{Q}_i$ is $q$; $\procName{bigDotExp}(\mat{\Phi},\{\mat{A}_i =
\mat{Q}_i\mat{Q}_i^\tr\}_{i=1}^n)$ computes $(1 \pm \vareps)$ approximations to
all $\exp{(\mat{\Phi})} \bigdot \mat{A}_i$ in $O(\kappa \log{m} \log
(\nicefrac{1}{\vareps}))$ depth and $O(\frac{1}{\vareps^2}(\kappa \log
(\nicefrac{1}{\epsilon}) p + q)\log{m})$ work.
\end{theorem}

The idea behind Theorem~\ref{thm:exp} is to approximate the matrix exponential
using a low-degree polynomial because evaluating matrix exponentials exactly is
costly.  For this, we will apply the following lemma, reproduced from Lemma~6
in~\cite{AroraK:stoc07}:
\begin{lemma}[\cite{AroraK:stoc07}]
\label{lem:expapprox}
If $\mat{B}$ is a PSD matrix such that $\norm[2]{\mat{B}} \leq \kappa$, then the
operator
\begin{align*}
  \widehat{\mat{B}} = \sum_{0 \leq i < k} \frac{1}{i!} \mat{B}^i \qquad \text{
    where } k = \max\{e^2\kappa, \ln(2\vareps^{-1})\}
\end{align*}
satisfies
\begin{align*}
(1  - \vareps)\exp{(\mat{B})}
\preceq \widehat{\mat{B}}
\preceq \exp{(\mat{B})}.
\end{align*}
\end{lemma}

\begin{proof}[of Theorem~\ref{thm:exp}]
  The given factorization of each $\mat{A}_i$ allows us to write
  $\exp{(\mat{\Phi})} \bigdot \mat{A}_i$ as the $2$-norm of a vector:
  \begin{align*}
    \exp{(\mat{\Phi})} \bigdot \mat{A}_i
    = & \trace{\exp{(\mat{\Phi})} \mat{Q}_i\mat{Q}_i^\tr} \\
    = & \trace{\mat{Q}_i^\tr \exp{(\tfrac{1}{2} \mat{\Phi})} \exp{(\tfrac{1}{2} \mat{\Phi})} \mat{Q}_i} \\
    = & \norm[2]{\exp{(\tfrac{1}{2} \mat{\Phi})} \mat{Q}_i}
  \end{align*}

  By Lemma \ref{lem:expapprox}, it suffices to evaluate $\widehat{\mat{B}}
  \bigdot \mat{A}_i$ where $\widehat{\mat{B}}$ is an approximation to $\mat{B} =
  \exp(\tfrac{1}{2}\mat{\Phi})$.  To further reduce the work, we can apply the
  Johnson-Lindenstrauss transformation~\cite{DasguptaG:rsa03,IndykM98} to reduce the
  length of the vectors to $O(\log{m})$; specifically, we find a
  $O(\frac{1}{\vareps^2}\log{m}) \times m$ Gaussian matrix $\mat{\Pi}$ and
  evaluate
\begin{align*}
\norm[2]{\mat{\Pi} \widehat{\mat{B}} \mat{Q}_i}
\end{align*}
Since $\mat{\Pi}$ only has $O(\frac1{\vareps^2}\log{m})$ rows, we can compute
$\mat{\Pi} \widehat{\mat{B}}$ using $O(\log{m})$ evaluations of
$\widehat{\mat{B}}$.  The work/depth bounds follow from doing each of the
evaluations of $\widehat{\mat{B}}\Pi_i$, where $\Pi_i$ denotes the $i$-th column
of $\mat{\Pi}$, and matrix-vector multiplies involving $\mat{\Phi}$ in parallel.
\end{proof}

% For completeness, we'll also prove Lemma~\ref{lem:expapprox}:
% \begin{proof}[of Lemma~\ref{lem:expapprox}]
%   Consider the $k$-term Taylor's expansion of $\exp(t)$:
% \begin{align}
%   \exp(t)
%   & = \sum_{0 \leq  i < k} \frac{1}{i!} t^i + \frac{\gamma}{k!}t^{k}
% \end{align}
% where $\gamma$ is a constant between $0$ and $2$.  This means that
% \begin{align}
% \exp(t)  - \frac{2}{k!}t^{k}
% \leq \sum_{0 \leq  i < k} \frac{1}{i!} t^i
% \leq \exp(t)
% \end{align}

% In the matrix setting, this becomes
% \begin{align}
% \exp{(\mat{M})} - \frac{2}{k!}\norm[2]{\mat{M}}^k\mident
% \preceq & \mat{B}
% \preceq \exp{(\mat{M})} \\
% \exp{(\mat{M})} - \frac{2}{k!}\kappa^k\mident
% \preceq & \mat{B}
% \preceq \exp{(\mat{M})}
% \end{align}

% Since $0 \preceq \mat{M}$, $\mident \preceq \exp{(\mat{M})}$, which gives:

% \begin{align}
% (1 - \frac{2}{k!}\kappa^k )\exp{(\mat{M})}
% \preceq & \mat{B}
% \preceq \exp{(\mat{M})}
% \label{eq:experror}
% \end{align}

% Consider substituting in Stirling's approximation for $k!$, aka.
% $k! \geq \sqrt{2 \pi k} (k/e)^k \geq (k/e)^k$ into
% $\frac{2}{k!}\kappa^{k}$:

% \begin{align}
% \frac{2}{k!}\kappa^{k}
% \leq & 2 \kappa^{k} \left( \frac{e}{k} \right) ^ {k} \nonumber \\
% = & 2 \left( \frac{e \kappa } {k}\right) ^ {k} \nonumber \\
% = & 2 \left( \frac{1} {e} \right) ^ {k}
% 	\qquad \text{Since $k \geq e^2 \kappa$}\nonumber \\
% \leq & \vareps
% 	\qquad \text{Since $k \geq \ln(2 \vareps^{-1})$}
% \end{align}

% Substituting this into Equation \ref{eq:experror} gives the result.

% \end{proof}

%%% Local Variables:
%%% mode: latex
%%% TeX-master: "paper"
%%% End:

%\input{packing-covering}
\section{Conclusion}
\label{sec:concl}

We presented a simple \NC parallel algorithm for packing SDPs that requires
$O(\frac1{\vareps^4} \log^4 n \log(\frac{1}{\vareps}))$ iterations, where each
iteration involves only simple matrix operations and computing the trace of the
product of a matrix exponential and a positive semidefinite matrix.  When a
positive SDP is given in a factorized form, we showed how the dot product with
matrix exponential can be implemented in nearly-linear work, leading to an
algorithm with $\otilde(m+n+q)$ work, where $n$ is the number of constraint
matrices, $m$ is the dimension of these matrices, and $q$ is the total number of
nonzero entries in the factorization.

Compared to the situation with positive LPs, the classification of positive SDPs
is much richer because packing/covering constraints can take many forms, either
as matrices (e.g. $\sum_{i=1}^n x_i \mat{A}_i \mleq \mident$ for packing,
$\sum_{i=1}^n x_i \mat{A}_i \mgeq \mident$ for covering) or as dot products
between matrices (e.g. $\mat{A}_i \bigdot \mat{Y} \leq 1$ for packing,
$\mat{A}_i \bigdot \mat{Y} \geq 1$ for covering).  The positive SDPs studied in
\cite{JainY:focs11} and our paper should be compared with the closely related
notion of covering SDPs studied by Iyengar et al~\cite{IyengarPS:swat10};
however, among the applications they examine, only the beamforming SDP
relaxation discussed in Section 2.2 of \cite{IyengarPS:swat10} falls completely
within the framework of packing SDPs as defined in \ref{eq:simplified-sdp}.
Problems such as \textsc{MaxCut} and \textsc{SparsestCut} require additional
matrix-based packing constraints.  We believe extending these algorithms to
solve mixed packing/covering SDPs is an interesting direction for future work.
%%% Local Variables:
%%% mode: latex
%%% TeX-master: "paper"
%%% End:

%\vspace{-0.1in}
\section*{Acknowledgments}
This work is partially supported by the National Science Foundation under grant
numbers CCF-1018463, CCF-1018188, and CCF-1016799 and by generous gifts from
IBM, Intel, and Microsoft.
Richard Peng was partly supported by a Microsoft Research PhD. Fellowship.

We thank the SPAA reviewers, as well as Guy Blelloch and Gary Miller
for suggestions that helped improve this paper.
While making this revision, we benefitted greatly from discussions with
Jon Kelner and Di Wang.

%We are grateful for all the fish.
% Comments for people we need to ack in the final version

%% Bibliography
\setlength{\bibsep}{1pt}
% \footnotesize % SPACE
\begin{small}
\bibliographystyle{alpha}
%\bibliography{../ref}

\end{small}
%\bibliographystyle{abbrvnat_noaddr} % SPACE
%\theendnotes % ENDNOTES
}{% !onlyAbstract
}

\begin{appendix}
\section{Normalized Positive SDPs}
\label{sec:norm-posit-sdps}
This is the same transformation that Jain and Yao presented~\cite{JainY:focs11};
we only present it here for easy reference.

Consider the primal program in \eqref{eq:covering}. It suffices to show that
it can be transformed into the following program without changing the optimal
value:
\begin{align}
  \label{eq:covering-renormalized}
  \begin{array}{l l l}
    \text{Minimize} & \trace{\mat{Z}}\\
    \text{Subject to:} & \mat{B}_i \bigdot \mat{Z}  \geq 1 &\text{ for } i = 1, \dots, m\\
    &\mat{Z}  \mgeq \mat{0},
  \end{array}
\end{align}

We can make the following assumptions without loss of generality: First, $b_i >
0$ for all $i = 1, \dots, m$ because if $b_i$ were $0$, we could have thrown it
away.  Second, all $\mat{A}_i$'s are the support of $\mat{C}$, or otherwise we
know that the corresponding dual variable must be set to $0$ and therefore can
be removed right away.  Therefore, we will treat $C$ as having a full-rank,
allowing us to define
\[
\mat{B}_i \eqdef \frac1{b_i} \mat{C}^{-1/2}\mat{A}_i \mat{C}^{-1/2}
\]
It is not hard to verify that the normalized program
\eqref{eq:covering-renormalized} has the same optimal value as the original SDP
\eqref{eq:covering}.

Note that if we're given factorization of $\mat{A}_i$ into
$\mat{Q}_i\mat{Q}_i^\tr$, then $\mat{B}_i$ can also be factorized as:

\[
\mat{B}_i = \frac1{b_i} (\mat{C}^{-1/2}\mat{Q}_i) (\mat{C}^{-1/2} \mat{Q}_i)^\tr
\]

Furthermore, it can be checked that the dual of the normalized program is
the same as the dual in Equation \ref{eq:simplified-sdp}.

%%% Local Variables:
%%% mode: latex
%%% TeX-master: "paper"
%%% End:

\end{appendix}

\end{document}